\documentclass[journal,letterpaper]{IEEEtran}

\usepackage[cmex10]{amsmath}
\usepackage{amssymb}
\usepackage{amsfonts}
\usepackage{amsthm}
\usepackage{bm}
\usepackage{graphicx}
\usepackage{cite}
\usepackage{pgfplots}
\usepgfplotslibrary{fillbetween}
\pgfplotsset{compat=1.18}
\usepackage{amsthm}
\newtheorem{theorem}{Theorem}

\theoremstyle{definition}

\begin{document}

\title{USAM: A Unified Safety–Age Metric for Timeliness in Heterogeneous IoT Systems}

\author{Mikael~Gidlund,~\IEEEmembership{Fellow,~IEEE}
\thanks{M. Gidlund is with Department of Computer and Electrical Engineering. Email: \texttt{mikael.gidlund@miun.se}}}

\maketitle

\begin{abstract}
Massive Internet-of-Things (IoT) deployments must simultaneously support monitoring, control, and safety-critical communication over shared wireless infrastructure. Classical timeliness metrics, such as Age of Information and its variants, quantify the freshness of received updates but do not account for deterministic safety timing requirements that arise in cyber-physical systems. Consequently, freshness-oriented metrics may indicate satisfactory performance even when worst-case timing guarantees required by functional safety standards are violated. This paper introduces the Unified Safety--Age Metric (USAM), a safety-aware timeliness metric that integrates information freshness, deadline reliability, and deterministic response-time feasibility into a single architecture-aware performance measure. We consider heterogeneous IoT traffic served by a gateway with intermittent receiver readiness and analyze system behavior in the ultra-sparse regime typical of massive machine-type communications. The analysis shows that, as device activity decreases, queueing delays become negligible and system timeliness becomes dominated by infrastructure readiness and deterministic response-time constraints. In this regime, feasibility is determined primarily by the receiver duty cycle rather than by average traffic load. Numerical results illustrate the safety-blindness of classical freshness metrics and demonstrate that USAM explicitly captures the feasibility boundary imposed by heterogeneous traffic requirements. The proposed framework provides a foundation for analyzing safety-aware communication architectures in large-scale IoT systems.
\end{abstract}


\section{Introduction}

\IEEEPARstart{M}{assive} Internet-of-Things (IoT) deployments envisioned for sixth-generation (6G) wireless systems will support extremely large populations of devices generating sporadic short packets for applications such as environmental monitoring, smart cities, precision agriculture, intelligent transportation, and industrial automation. In such systems, the number of connected devices may reach millions per square kilometer. In contrast, only a small fraction is active at any given time due to the sporadic nature of machine-type traffic. This behavior is commonly described through an activity factor $\rho$ representing the fraction of devices generating traffic at a given time. Massive IoT systems, therefore, operate in an ultra-sparse regime in which the aggregate traffic load remains small even though the total number of connected devices is extremely large~\cite{GidlundEPB}.

The ultra-sparse operating regime changes the interaction between traffic dynamics and communication infrastructure. In particular, the readiness of the receiver chain at the base station determines whether incoming packets can be processed immediately or must wait until the reception pipeline becomes available. In many practical IoT deployments, the receiver may not operate continuously due to energy constraints, hardware limitations, or system-level scheduling policies~\cite{GidlundEPB}. As a result, the effective packet processing capability of the infrastructure depends not only on the nominal service rate of the receiver but also on the fraction of time during which the reception chain is active. The operational duty cycle of the receiver, therefore, directly influences queueing behavior and achievable timeliness guarantees in ultra-sparse IoT systems.

The timeliness of information delivery in communication networks has been studied extensively through the framework of \emph{Age of Information (AoI)}, which measures the time elapsed since the generation of the most recently received update~\cite{KaulINFOCOM,YatesAoI}. AoI has become a standard metric for monitoring systems and has been analyzed under a wide range of queueing models, scheduling policies, and sampling strategies~\cite{AoISurvey,KostaAoI}. These studies established the foundations of freshness-aware communication and highlighted the limitations of classical delay metrics in settings where information updates are generated continuously.

Several extensions of AoI have been proposed to capture additional aspects of information timeliness. Examples include synchronization-oriented formulations such as the \emph{Age of Synchronization (AoS)}~\cite{ChenAoS}, correctness-aware metrics such as the \emph{Age of Incorrect Information (AoII)}~\cite{AoII}, control-oriented formulations such as the \emph{Age of Control (AoC)}, which relates communication delay to the performance of networked control systems~\cite{AoC}, and application-aware metrics such as the \emph{Value of Information (VoI)}~\cite{ChiariottiVoI}. These developments illustrate the diversity of timeliness metrics tailored to specific system objectives.

Recent studies emphasize that timing should be treated as a fundamental design dimension of future wireless systems rather than merely a byproduct of delay performance~\cite{PopovskiTime6G}. In cyber-physical IoT systems, the usefulness of transmitted information depends on the interaction between sensing, communication, and actuation processes. Communication systems must therefore be evaluated with respect to how effectively they support time-sensitive system behavior.

Many emerging IoT applications operate in cyber-physical environments where communication delays directly affect physical processes. Examples include smart infrastructure, autonomous transportation systems, robotics, and industrial automation. Such systems must support heterogeneous traffic classes ranging from low-rate monitoring streams to latency-critical control loops and safety functions. In many industrial settings, these requirements are formalized through functional safety standards that impose deterministic timing guarantees.

The International Electrotechnical Commission (IEC) standard 61784-3 specifies communication profiles for safety fieldbuses and defines \emph{Safety Integrity Levels (SILs)} together with associated \emph{Safety Functional Response Time (SFRT)} constraints~\cite{IEC61784}. Compliance with these standards requires that safety functions be executed within strictly bounded time intervals. Wireless sensor--actuator networks used in factory automation must therefore guarantee deterministic cycle times, low jitter, and extremely high reliability to support safety functions involving human--machine interaction~\cite{DoebbertIOLW}. These constraints impose deterministic limits on end-to-end communication delays.

Providing such guarantees remains challenging in wireless communication systems. Technologies such as Time-Sensitive Networking (TSN) enable bounded latency scheduling in wired Ethernet networks, but extending these concepts to heterogeneous wireless environments is significantly more complex due to channel variability, contention-based access, and stochastic delays. Similarly, ultra-reliable low-latency communication (URLLC) services envisioned for 5G and future 6G systems primarily focus on probabilistic latency guarantees rather than deterministic safety response times. These limitations highlight the need for system-level metrics that jointly capture stochastic timeliness and deterministic safety feasibility in heterogeneous IoT deployments.

This paper introduces the \emph{Unified Safety--Age Metric (USAM)}, a system-level feasibility metric that captures the interaction between information freshness, deadline reliability, and deterministic safety timing in heterogeneous IoT deployments. The metric integrates three complementary dimensions of timeliness: a freshness component capturing synchronization quality across monitoring sources, a reliability component reflecting deadline satisfaction and Safety Integrity Level requirements, and a deterministic component capturing worst-case response-time feasibility relative to the Safety Functional Response Time. The objective of USAM is not to introduce another stochastic timeliness metric, but to provide an architectural feasibility measure that determines whether a communication infrastructure can simultaneously support freshness, reliability, and deterministic safety requirements in mixed-criticality IoT systems.

The analysis shows that massive IoT systems operating in the ultra-sparse regime exhibit a structural transition in how timeliness is governed. When the activity factor is small, queueing delays become negligible, and system behavior is dominated by infrastructure readiness and deterministic safety constraints rather than by traffic load. This observation leads to a simple infrastructure design principle: the duty cycle of the receiver infrastructure must exceed a minimum threshold determined jointly by queue-stability and worst-case response-time feasibility conditions.

The main contributions of this paper are summarized as follows:

\begin{itemize}

\item We introduce the Unified Safety--Age Metric (USAM), which integrates information freshness, deadline reliability, and deterministic response-time feasibility into a unified formulation for heterogeneous IoT systems.

\item We derive asymptotic expressions for the components of USAM in the ultra-sparse operating regime under intermittent receiver operation.

\item We characterize the feasibility conditions under which freshness, reliability, and deterministic safety constraints can be simultaneously satisfied and derive an explicit expression for the minimum safe receiver duty cycle.

\item We show that classical timeliness metrics such as Age of Information arise as special cases of the proposed framework when deterministic safety constraints are not considered.

\end{itemize}

The remainder of this paper is organized as follows. Section~II introduces the system model for heterogeneous IoT traffic under intermittent service operation. Section~III defines the Unified Safety--Age Metric and its components. Section~IV analyzes asymptotic feasibility in the ultra-sparse regime. Section~V presents the numerical evaluation. Section~VI concludes the paper.
\section{System Model}

We consider a heterogeneous Internet-of-Things (IoT) system in which a large population of distributed devices communicates with a common gateway or base station (BS) through a shared wireless infrastructure. The system supports multiple traffic classes with heterogeneous timing requirements, reflecting the coexistence of monitoring, control, and safety-related communication typical of modern IoT deployments. Monitoring streams coexist with control traffic and safety alarms while sharing the same communication resources. The system abstraction used throughout the paper is illustrated in Fig.~\ref{fig:system_model}.

\begin{figure}[t]
\centering
\includegraphics[width=0.95\linewidth]{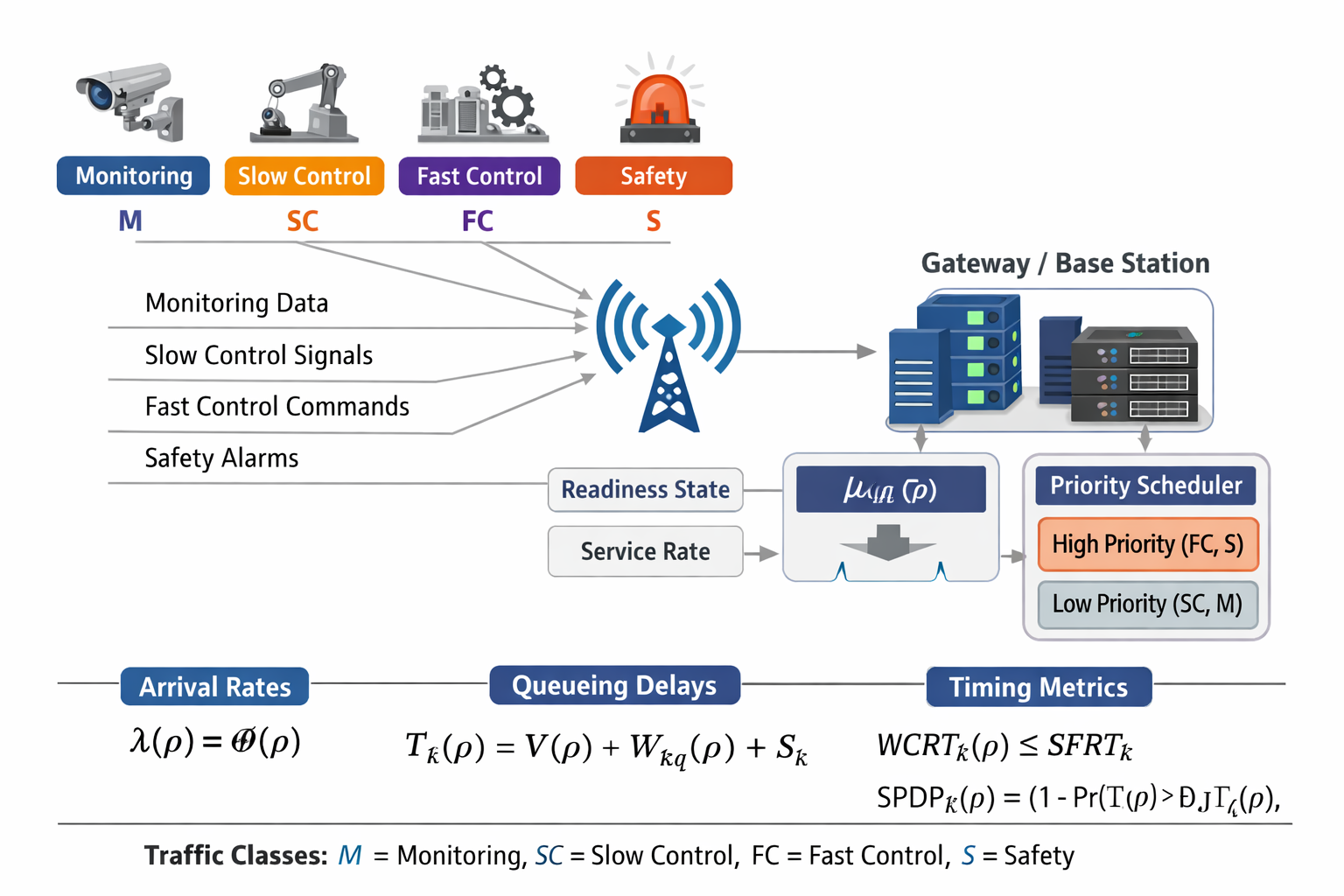}
\caption{System model of the heterogeneous IoT uplink. Multiple traffic classes share a common wireless infrastructure and are processed by a gateway with intermittent readiness. Packets are served by a fixed-priority scheduler with class-dependent timing constraints.}
\label{fig:system_model}
\end{figure}

Let
\begin{equation}
\mathcal{K}=\{M,SC,FC,S\}
\end{equation}
denote the set of traffic classes corresponding to monitoring ($M$), slow control ($SC$), fast control ($FC$), and safety ($S$) traffic. These classes reflect the hierarchy of timing requirements commonly encountered in industrial automation and cyber--physical systems, where low-rate monitoring streams coexist with time-critical control loops and safety-related messages.

Packets of class $k\in\mathcal{K}$ arrive according to a stationary arrival process with rate $\lambda_k(\rho)$, where $\rho$ denotes the activity factor of the system, defined as the long-term fraction of devices generating traffic. The aggregate arrival rate is therefore
\begin{equation}
\lambda(\rho)=\sum_{k\in\mathcal{K}}\lambda_k(\rho).
\end{equation}

The analytical development does not rely on a specific arrival distribution beyond stationarity and finite first and second moments. In the numerical evaluation, arrivals are modeled as Poisson processes, which yields an $M/G/1$ queue with intermittent service availability.

We focus on the ultra-sparse operating regime characteristic of massive IoT systems, in which only a small fraction of devices is active at any given time. In this regime the aggregate traffic load scales proportionally with the activity factor,
\begin{equation}
\lambda(\rho)=\Theta(\rho), \qquad \rho\rightarrow0 .
\end{equation}

At the gateway, the reception chain is modeled as a single-server queue with intermittent service availability. When the receiver is active, packets are processed with rate $\mu>0$, corresponding to the nominal service time
\begin{equation}
S=\frac{1}{\mu}.
\end{equation}

Service requirements may differ across traffic classes. Let $S_k$ denote the service time of a class-$k$ packet and assume that $S_k$ is a random variable with bounded support
\begin{equation}
0\le S_k\le C_k<\infty ,
\end{equation}
where $C_k$ represents the maximum service requirement of class $k$.

The receiver may not operate continuously. Instead, the reception chain may be intermittently available due to energy constraints, hardware limitations, or system-level duty-cycling policies. Let
\begin{equation}
R(t)\in\{0,1\}
\end{equation}
denote the receiver's readiness state, where $R(t)=1$ indicates that the gateway can process packets. The long-term receiver duty cycle is defined as
\begin{equation}
\delta=\lim_{T\to\infty}\frac{1}{T}\int_0^T R(t)\,dt .
\end{equation}

When the receiver is active, packets are served at a rate $\mu$, resulting in the effective service rate
\begin{equation}
\mu_{\mathrm{eff}}=\mu\,\delta .
\end{equation}

Packets arriving while the receiver is inactive experience an activation delay before service begins. Let $V(\delta)$ denote this activation latency. In the worst case, a packet arrives immediately after the receiver has deactivated, which yields the bound
\begin{equation}
V(\delta)\le V_{\max}(1-\delta).
\label{eq:Vmax_delta}
\end{equation}

Let $W_k^{(q)}(\rho)$ denote the queueing delay experienced by a class-$k$ packet excluding activation latency. The response time of class $k$ is therefore
\begin{equation}
T_k(\rho)=V(\delta)+W_k^{(q)}(\rho)+S_k .
\end{equation}

Packets are served according to a fixed-priority preemptive scheduling discipline reflecting the timing hierarchy of industrial control systems,
\begin{equation}
S \succ FC \succ SC \succ M .
\end{equation}

For a class $k$, let
\begin{equation}
\mathrm{hp}(k)=\{i\in\mathcal{K}\mid \text{priority}(i)>\text{priority}(k)\}
\end{equation}
denote the set of traffic classes with higher priority.

For safety analysis, we consider the worst-case response time (WCRT), defined as the maximum delay a packet may experience under the scheduling policy. Under fixed-priority scheduling, the deterministic bound for class $k$ is
\begin{equation}
B_k(\delta)=V_{\max}(1-\delta)+\sum_{i\in\mathrm{hp}(k)}C_i+C_k .
\label{eq:Bk}
\end{equation}

Safety-critical traffic must additionally satisfy deterministic Safety Functional Response Time (SFRT) constraints defined in functional safety standards such as IEC~61784-3. The SFRT specifies the maximum allowable time between the occurrence of a hazardous event and the completion of the corresponding protective action. Safety feasibility, therefore, requires
\begin{equation}
B_k(\delta)\le \mathrm{SFRT}_k .
\end{equation}
\section{Unified Safety--Age Metric}

The system model introduced in the previous section captures heterogeneous IoT traffic composed of monitoring, control, and safety-related communication flows. These traffic classes impose fundamentally different timing requirements. Monitoring applications primarily require information freshness, control applications require bounded response times to preserve closed-loop stability, and safety-critical functions must satisfy deterministic timing constraints imposed by functional safety standards. Existing timeliness metrics address these requirements only partially. Freshness metrics such as Age of Information (AoI) quantify the staleness of received updates, while deadline-based reliability metrics evaluate the probability that packets meet specified delay constraints. Control-oriented formulations relate communication delay to control-loop performance. However, these formulations remain fundamentally stochastic because they evaluate expectations or probabilities of delay distributions.

Functional safety constraints differ in nature. Industrial safety standards characterize the timing of protective actions through the Safety Functional Response Time (SFRT), defined as the maximum allowable time between the occurrence of a hazardous event and the completion of the corresponding safety function~\cite{IEC61784}. In networked cyber-physical systems this interval includes sensing, communication, processing, and actuation stages. Safe operation therefore requires deterministic guarantees on the maximum communication delay. In real-time systems such guarantees are expressed through the Worst-Case Response Time (WCRT), which represents the maximum delay experienced by a task or message under all admissible scheduling conditions. Using the response-time bound derived in the system model, safety feasibility requires
\begin{equation}
B_k(\delta) \le \mathrm{SFRT}_k .
\end{equation}

This deterministic constraint differs fundamentally from stochastic delay metrics commonly used in communication networks. Expectation-based measures such as $\mathbb{E}[T]$ or $\mathbb{E}[\Delta]$ may remain small even when the delay distribution admits realizations exceeding the safety limit. Functional safety instead requires the strict condition
\begin{equation}
P(T_k > \mathrm{SFRT}_k)=0 ,
\end{equation}
which cannot be guaranteed by expectation-based freshness metrics alone. Consequently, communication systems optimized according to stochastic timeliness metrics may appear to operate satisfactorily while violating deterministic safety constraints.

Figure~\ref{fig:capability_map} summarizes the structural difference between stochastic timeliness metrics and deterministic safety feasibility. AoI- and PAoI-type metrics remain safety-blind because they quantify information staleness without incorporating worst-case timing guarantees. VoI and AoC incorporate additional system semantics but remain below the deterministic safety boundary. The proposed Unified Safety--Age Metric extends the metric space by explicitly incorporating deterministic response-time feasibility.

\begin{figure}[t]
\centering
\includegraphics[width=0.92\linewidth]{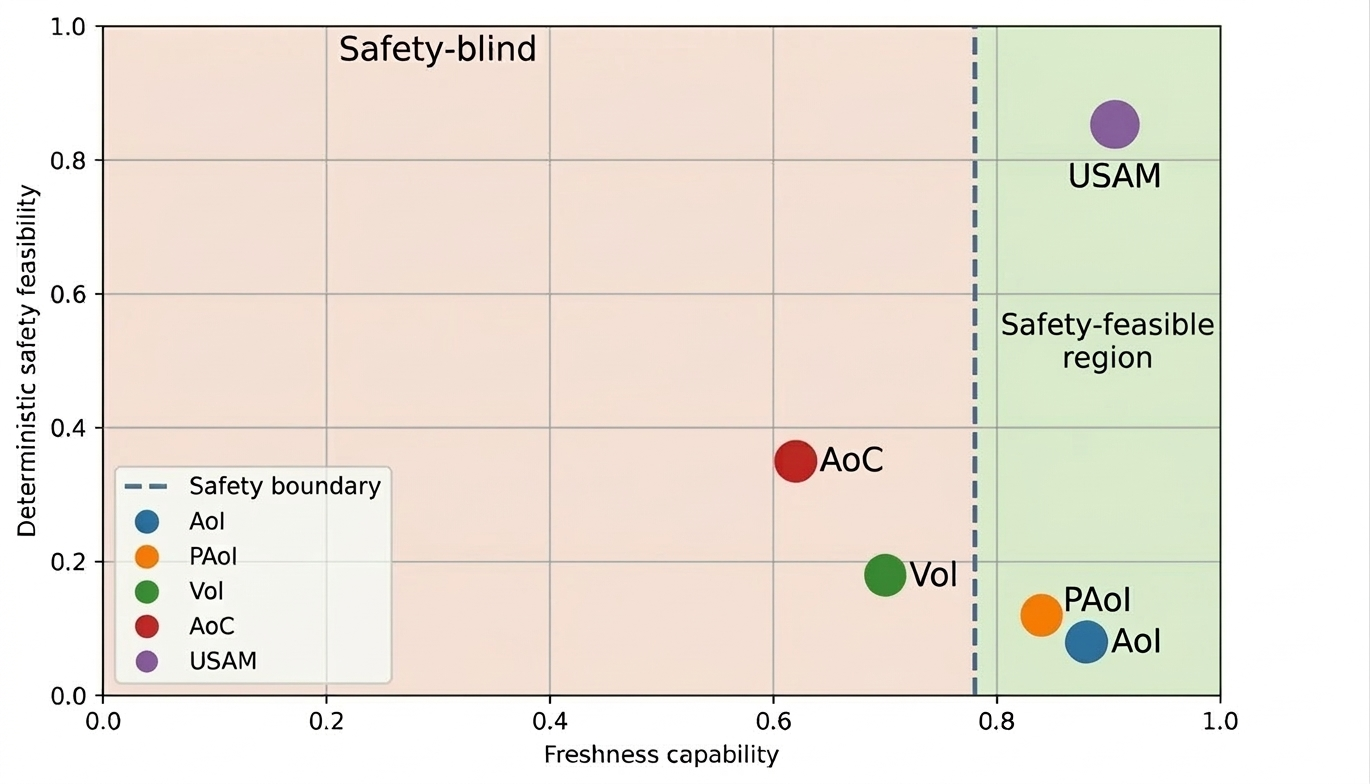}
\caption{Conceptual capability space of representative timeliness metrics. Freshness-oriented metrics operate in a safety-blind region because they do not incorporate deterministic safety constraints. The proposed Unified Safety--Age Metric (USAM) operates in the safety-feasible region by jointly capturing freshness, reliability, and deterministic safety timing.}
\label{fig:capability_map}
\end{figure}

To jointly capture freshness, reliability, and deterministic safety timing, we introduce the \emph{Unified Safety--Age Metric (USAM)}. Unlike classical age-based formulations that quantify average timeliness, USAM evaluates whether a communication infrastructure can simultaneously satisfy freshness, reliability, and deterministic timing requirements. The metric, therefore, acts as a system-level feasibility measure for heterogeneous IoT deployments.

Monitoring traffic is evaluated through the synchronization quality of received updates. Let $\mathcal{K}_M \subseteq \mathcal{K}$ denote the set of monitoring classes and let $\mathbb{E}[\mathrm{AoS}(\rho)]$ denote the average Age of Synchronization at the gateway. Let $\Delta_{\mathrm{tar}}>0$ denote a target synchronization level representing the maximum acceptable update age. The normalized freshness factor is defined as
\begin{equation}
F(\rho)=
\min\!\left(
1,\,
\frac{\Delta_{\mathrm{tar}}}{\mathbb{E}[\mathrm{AoS}(\rho)]}
\right).
\end{equation}

Delay-sensitive traffic classes require reliable packet delivery within specified deadlines. Let $D_k>0$ denote the deadline associated with class $k$. Recall that the Safe Packet Delivery Probability (SPDP) of class $k$ is defined as
\begin{equation}
\mathrm{SPDP}_k(\rho)=
\left(
1-\Pr\{T_k(\rho)>D_k\}
\right)\Gamma_k ,
\end{equation}
where $\Gamma_k\in(0,1]$ represents a residual safety factor reflecting Safety Integrity Level requirements and other non-communication failure modes. Since heterogeneous IoT systems must satisfy the most stringent reliability requirement among all traffic classes, the reliability component of the unified metric is defined as
\begin{equation}
R(\rho)=
\min_{k\in\mathcal{K}}
\mathrm{SPDP}_k(\rho).
\end{equation}

Safety-critical operation additionally requires deterministic timing guarantees. Let $\mathrm{SFRT}^{*}>0$ denote a reference Safety Functional Response Time used for normalization. The deterministic safety component is defined as
\begin{equation}
S(\rho)=
\exp\!\left(
-
\frac{\max_{k\in\mathcal{K}} B_k(\delta)}
{\mathrm{SFRT}^{*}}
\right),
\end{equation}
where $B_k(\delta)$ denotes the worst-case response-time bound of class $k$.

Since $F(\rho)$, $R(\rho)$, and $S(\rho)$ are normalized quantities in the interval $[0,1]$, the unified metric is defined as the weighted geometric aggregation
\begin{equation}
\Psi(\rho)=
F(\rho)^{w_1}
R(\rho)^{w_2}
S(\rho)^{w_3},
\label{eq:USAM}
\end{equation}
where $w_1,w_2,w_3 \ge 0$ are design parameters controlling the relative importance of freshness, reliability, and deterministic safety timing.

The multiplicative structure reflects the joint nature of the requirements: freshness, reliability, and deterministic safety must simultaneously hold for safe system operation. Consequently, degradation in any individual dimension reduces the overall metric value. The use of the minimum reliability across traffic classes ensures that the most demanding traffic class governs deadline feasibility, while the maximum worst-case response time captures the dominant deterministic safety bottleneck.

The Unified Safety--Age Metric satisfies several structural properties. Since each component lies in the interval $[0,1]$, the metric is bounded as $0 \le \Psi(\rho) \le 1$. Furthermore, for nonnegative weights the metric is monotone with respect to each component, implying that improvements in freshness, reliability, or deterministic safety timing cannot reduce the overall metric value.

The unified formulation also clarifies the relationship to existing timeliness metrics. If $w_2=w_3=0$, the metric reduces to a normalized freshness measure determined by the Age of Synchronization. If $w_3=0$, the formulation captures freshness and probabilistic deadline reliability but ignores deterministic safety timing. Conversely, when monitoring traffic is absent, the formulation reduces to a reliability--safety metric governed by deadline reliability and worst-case response-time feasibility. These special cases illustrate that widely used timeliness metrics correspond to partial projections of the unified formulation.

The next section analyzes the asymptotic behavior of $\Psi(\rho)$ under intermittent service in the ultra-sparse operating regime.
\section{Asymptotic Feasibility in the Ultra-Sparse Regime}

This section analyzes the behavior of the Unified Safety--Age Metric introduced in Section~III under the ultra-sparse operating regime of massive IoT systems. In such systems, only a small fraction of devices are active at any given time, and the aggregate traffic load therefore scales proportionally with the activity factor $\rho$. The objective is to determine how the freshness, reliability, and deterministic safety components of the unified metric behave as $\rho \rightarrow 0$.

Let $\lambda_k(\rho)$ denote the arrival rate of class $k$ and define the aggregate arrival rate
\begin{equation}
\lambda(\rho)=\sum_{k\in\mathcal{K}}\lambda_k(\rho).
\end{equation}

In the ultra-sparse regime, the traffic load satisfies
\begin{equation}
\lambda(\rho)=\Theta(\rho), \qquad \rho\rightarrow0 .
\end{equation}

The gateway receiver operates with intermittent readiness characterized by duty cycle $\delta$. When the receiver is active, packets are served with rate $\mu$, yielding the effective service rate
\begin{equation}
\mu_{\mathrm{eff}}=\mu\delta .
\end{equation}

Define the system utilization
\begin{equation}
\varrho(\rho)=\frac{\lambda(\rho)}{\mu_{\mathrm{eff}}}
=\frac{\lambda(\rho)}{\mu\delta}.
\end{equation}

Since $\lambda(\rho)=\Theta(\rho)$ and $\delta$ is independent of $\rho$, it follows that
\begin{equation}
\varrho(\rho)\rightarrow0
\qquad \text{as}\qquad \rho\rightarrow0 .
\end{equation}

The vanishing utilization places the system in the light-load regime of the single-server queue. Under standard light-traffic asymptotics for queues with finite service-time moments, the expected queueing delay satisfies
\begin{equation}
\mathbb{E}[W_k^{(q)}(\rho)] = O(\varrho(\rho)),
\end{equation}
and therefore
\begin{equation}
\mathbb{E}[W_k^{(q)}(\rho)]\rightarrow0
\qquad \text{as}\qquad \rho\rightarrow0 .
\end{equation}

Hence queueing delays vanish asymptotically and the response time becomes dominated by activation latency and service time.

Monitoring traffic is evaluated through the Age of Synchronization. Let $\mathcal{K}_M$ denote the set of monitoring classes and write
\begin{equation}
\lambda_k(\rho)=\rho\,\lambda_s\,\pi_k ,
\end{equation}
where $\lambda_s$ denotes the arrival-rate scale and $\pi_k$ represents the traffic mix.

Under the service model of Section~II the Age of Synchronization can be approximated as
\begin{equation}
\mathbb{E}[\mathrm{AoS}(\rho)]
=
\max_{k\in\mathcal{K}_M}
\left(
\frac{1}{\lambda_k(\rho)}
+
\frac{\mathbb{E}[W_k^{(q)}(\rho)]}{2}
+
\frac{1}{\mu\delta}
\right).
\end{equation}

Since $\mathbb{E}[W_k^{(q)}(\rho)]=O(\varrho(\rho))$ and $\varrho(\rho)=\Theta(\rho)$, the queueing contribution vanishes relative to the update-generation term $1/\lambda_k(\rho)$. Consequently
\begin{equation}
\mathbb{E}[\mathrm{AoS}(\rho)]
=
\max_{k\in\mathcal{K}_M}
\left(
\frac{1}{\rho\lambda_s\pi_k}
+
\frac{1}{\mu\delta}
\right)
(1+o(1)).
\end{equation}

Thus, in ultra-sparse operation the freshness component is governed by monitoring update intensity and receiver duty cycle rather than by queueing delay.

The response time of class $k$ is
\begin{equation}
T_k(\rho)=V(\delta)+W_k^{(q)}(\rho)+S_k .
\end{equation}

Since $W_k^{(q)}(\rho)\rightarrow0$, deadline violations vanish asymptotically whenever deterministic slack exists. Specifically, if
\begin{equation}
D_k>V_{\max}(1-\delta)+C_k ,
\end{equation}
then
\begin{equation}
\Pr\{T_k(\rho)>D_k\}\rightarrow0
\qquad \text{as}\qquad \rho\rightarrow0 .
\end{equation}

Consequently
\begin{equation}
\mathrm{SPDP}_k(\rho)\rightarrow\Gamma_k .
\end{equation}

Unlike queueing delay, deterministic safety timing does not vanish in the ultra-sparse limit. From the activation model in Section~II the worst-case response-time bound is
\begin{equation}
B_k(\delta)=V_{\max}(1-\delta)+\sum_{i\in\mathrm{hp}(k)}C_i+C_k .
\end{equation}

Defining
\begin{equation}
B_k^{(0)}=\sum_{i\in\mathrm{hp}(k)}C_i+C_k
\end{equation}
gives
\begin{equation}
B_k(\delta)=V_{\max}(1-\delta)+B_k^{(0)} .
\end{equation}

This bound is independent of traffic load and therefore remains a structural constraint even when $\rho\rightarrow0$.

\begin{theorem}[Infrastructure-Dominated Timeliness]

Consider the heterogeneous IoT system defined in Section~II and the Unified Safety--Age Metric~\eqref{eq:USAM}. If $\varrho(\rho)\rightarrow0$ as $\rho\rightarrow0$, then
\begin{equation}
\Psi(\rho)=F(\rho)^{w_1}R(\rho)^{w_2}S(\rho)^{w_3}.
\end{equation}

As $\rho\rightarrow0$ the metric converges to
\begin{equation}
\Psi(\rho)
\rightarrow
F_0^{w_1}
(\Gamma_{\min})^{w_2}
\exp
\left(
-w_3\frac{\max_k B_k(\delta)}{\mathrm{SFRT}^{*}}
\right),
\end{equation}

where
\begin{equation}
F_0=
\min
\left(
1,
\frac{\Delta_{\mathrm{tar}}}
{\max_{k\in\mathcal{K}_M}
\left(
\frac{1}{\rho\lambda_s\pi_k}+\frac{1}{\mu\delta}
\right)}
\right)
\end{equation}
and $\Gamma_{\min}=\min_k\Gamma_k$.
\end{theorem}

\begin{proof}

As $\rho\rightarrow0$, the system utilization satisfies $\varrho(\rho)\rightarrow0$, which implies $\mathbb{E}[W_k^{(q)}(\rho)]\rightarrow0$. Substituting this result into the AoS expression yields the asymptotic freshness term determined by monitoring intensity and receiver duty cycle.

If deterministic slack exists, deadline violations vanish and $\Pr\{T_k(\rho)>D_k\}\rightarrow0$, implying $\mathrm{SPDP}_k(\rho)\rightarrow\Gamma_k$. Consequently the reliability component converges to $\Gamma_{\min}=\min_k\Gamma_k$.

The deterministic safety component depends only on the worst-case response-time bound $B_k(\delta)$, which is independent of $\rho$. Substituting these limits into the definition of $\Psi(\rho)$ yields the stated result.

\end{proof}

\textit{Implication (deterministic safety insensitivity of stochastic metrics):}
Metrics that depend solely on stochastic moments of the delay process, such as AoI or PAoI, remain insensitive to deterministic safety feasibility. In particular, it is possible that
\[
\mathbb{E}[T_k] \ll \mathrm{SFRT}_k
\]
while
\[
\Pr\{T_k > \mathrm{SFRT}_k\} > 0 ,
\]
which violates deterministic safety timing. Consequently freshness-oriented metrics may indicate satisfactory performance even when the system operates in a safety-infeasible regime.

The analysis also yields explicit feasibility conditions for the receiver duty cycle. Queue stability requires
\begin{equation}
\delta>\frac{\rho\lambda_s}{\varepsilon\mu}.
\end{equation}

Evaluating this bound at $\rho_{\max}$ yields
\begin{equation}
\delta_{\mathrm{queue}}
=
\frac{\rho_{\max}\lambda_s}{\varepsilon\mu}.
\end{equation}

Deterministic safety feasibility requires
\begin{equation}
\delta_{\mathrm{wcrt}}
=
\max_{k\in\mathcal{K}}
\max
\left\{
1-\frac{D_k-B_k^{(0)}}{V_{\max}},0
\right\}.
\end{equation}

\begin{theorem}[Safety Feasibility Threshold]

The minimum receiver duty cycle that satisfies both queue stability and deterministic safety constraints is
\begin{equation}
\delta_{\mathrm{safe}}
=
\max
\left(
\delta_{\mathrm{wcrt}},
\alpha\delta_{\mathrm{queue}}
\right),
\end{equation}
where $\alpha\ge1$ accounts for burstiness and channel impairments.

\end{theorem}

The system is infeasible for $\delta<\delta_{\mathrm{safe}}$. For $\delta\ge\delta_{\mathrm{safe}}$ both queue stability and deterministic safety timing are satisfied, and the Unified Safety--Age Metric remains strictly positive. The resulting feasibility regions in the $(\delta,\rho)$ plane are illustrated in Fig.~\ref{fig:phase_diagram}.
\section{Numerical Evaluation}

We evaluate the behavior of the Unified Safety--Age Metric (USAM) and compare it with representative timeliness metrics including Age of Information (AoI), Peak Age of Information (PAoI), Value of Information (VoI), and Age of Control (AoC). The objective is to illustrate how deterministic safety constraints interact with information freshness and deadline reliability in heterogeneous IoT systems and how these interactions determine the feasible operating region of the communication infrastructure.

The numerical setting follows the system abstraction introduced in Section~II. The uplink reception chain is modeled as a single-server queue with intermittent service availability induced by receiver duty-cycle control, corresponding to an $M/G/1$ queue with vacations. Packet processing times at the gateway lie in the range $S\in[0.05,0.2]$\,ms, consistent with digital front-end processing delays observed in industrial wireless gateways. The class-dependent parameters $C_k$ represent conservative worst-case service requirements including processing time, protocol overhead, and potential retransmissions. In the evaluation we use $C_M=0.20$\,ms, $C_{SC}=0.15$\,ms, $C_{FC}=0.12$\,ms, and $C_S=0.10$\,ms. Class deadlines reflect typical monitoring, control, and safety-cycle requirements with $D_M=80$\,ms, $D_{SC}=20$\,ms, $D_{FC}=8$\,ms, and $D_S=5$\,ms. The maximum activation latency is set to $V_{\max}=1.5$\,ms.

Unless otherwise stated, the USAM weights are chosen as $w_1=0.34$, $w_2=w_3=0.33$, assigning approximately equal importance to freshness, reliability, and deterministic safety timing. The safety normalization constant is $\mathrm{SFRT}^{*}=6$\,ms, representing a typical safety response-time scale in industrial control systems. Residual safety factors are set to $\Gamma_k=1$ for all traffic classes such that reliability is determined solely by communication-induced deadline violations.

Traffic is generated by a large population of devices with per-unit arrival-rate scale $\lambda_s=120$\,s$^{-1}$. To examine robustness with respect to mixed-criticality composition, the following traffic mixtures are considered:
\begin{align}
\pi^{(A)} &= \{0.95,0.05,0,0\}, \\
\pi^{(B)} &= \{0.25,0.25,0.25,0.25\}, \\
\pi^{(C)} &= \{0.70,0.15,0.10,0.05\}, \\
\pi^{(D)} &= \{0.50,0.25,0.10,0.15\}.
\end{align}

These correspond to a monitoring-dominated baseline, a balanced mixed-criticality configuration, a representative industrial IoT configuration, and a safety-intensive configuration, respectively. The aggregate activity factor $\rho$ represents the fraction of devices generating traffic at a given time and therefore captures the ultra-sparse operating regime typical of massive IoT deployments. The figures shown below correspond to the representative industrial IoT configuration $\pi^{(C)}$. The same qualitative behavior was observed across the other traffic mixtures and weight selections.

From the feasibility analysis in Section~IV, the queue-stability duty-cycle bound is
\[
\delta_{\mathrm{queue}}=\frac{\rho_{\max}\lambda_s}{\varepsilon\mu}.
\]
Using $\rho_{\max}=0.28$, $\varepsilon=0.2$, and $\mu=10^3$\,s$^{-1}$ yields $\delta_{\mathrm{queue}}=0.168$. For the considered parameters, the deterministic safety constraint does not impose an additional restriction because the available response-time margin absorbs the worst-case activation delay. Incorporating the design margin $\alpha=1.37$, therefore, gives the minimum feasible receiver duty cycle
\[
\delta_{\mathrm{safe}}=0.23.
\]

The effect of receiver duty cycle on the evaluated metrics is shown in Fig.~\ref{fig:safety_cliff}. The dashed vertical line indicates the safety feasibility threshold $\delta_{\mathrm{safe}}$. As the duty cycle decreases, AoI and PAoI increase smoothly because reduced receiver availability increases packet waiting time and therefore the age of received updates. VoI and AoC exhibit similar behavior since larger communication delays reduce update relevance and degrade control-loop performance.

These metrics depend on stochastic properties of the delay process and therefore do not capture deterministic response-time feasibility. Consequently, a system optimized for freshness may still violate hard safety timing constraints.

USAM incorporates both deadline reliability and deterministic response-time feasibility. As the receiver duty cycle approaches the feasibility boundary, the deterministic safety component becomes dominant because the worst-case response-time bound approaches the Safety Functional Response Time limit. Expectation-based metrics vary gradually with the mean delay, whereas deterministic constraints introduce a structural transition. As a result, USAM exhibits sharp degradation near the feasibility boundary, revealing a safety limit that is invisible to classical freshness-oriented metrics.

\begin{figure}[t]
\centering
\includegraphics[width=0.92\linewidth]{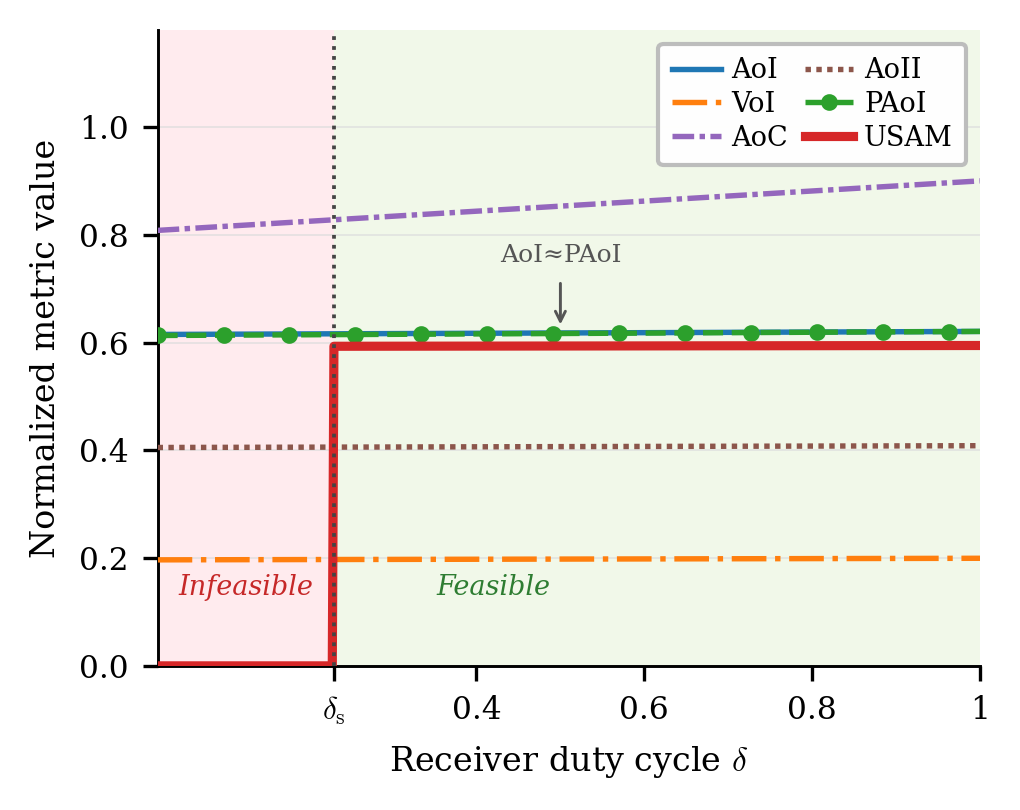}
\caption{Metric values as a function of receiver duty cycle $\delta$. The dashed vertical line marks the safety feasibility threshold $\delta_{\mathrm{safe}}$. Freshness-oriented metrics degrade smoothly as receiver availability decreases, whereas USAM captures the deterministic safety boundary.}
\label{fig:safety_cliff}
\end{figure}

The influence of traffic load is illustrated in Fig.~\ref{fig:rho_sweep}, where the metrics are shown as functions of the aggregate activity factor $\rho$ for a fixed receiver duty cycle. The dashed vertical line indicates the mixed-criticality feasibility threshold $\rho_{\mathrm{safe}}$. Increasing traffic load raises queueing delays and gradually degrades freshness-oriented metrics such as AoI and PAoI. VoI and AoC exhibit similar gradual degradation due to increased communication delay.

These curves do not exhibit a structural transition that indicates infeasible operation. Even when $\rho>\rho_{\mathrm{safe}}$, the classical metrics still report moderate values. In contrast, USAM incorporates reliability and deterministic timing constraints. As the activity factor approaches the feasibility threshold, the reliability and safety components dominate the metric. Beyond this point, the metric collapses, indicating that the system cannot simultaneously satisfy freshness, reliability, and deterministic safety requirements.

\begin{figure}[t]
\centering
\includegraphics[width=0.92\linewidth]{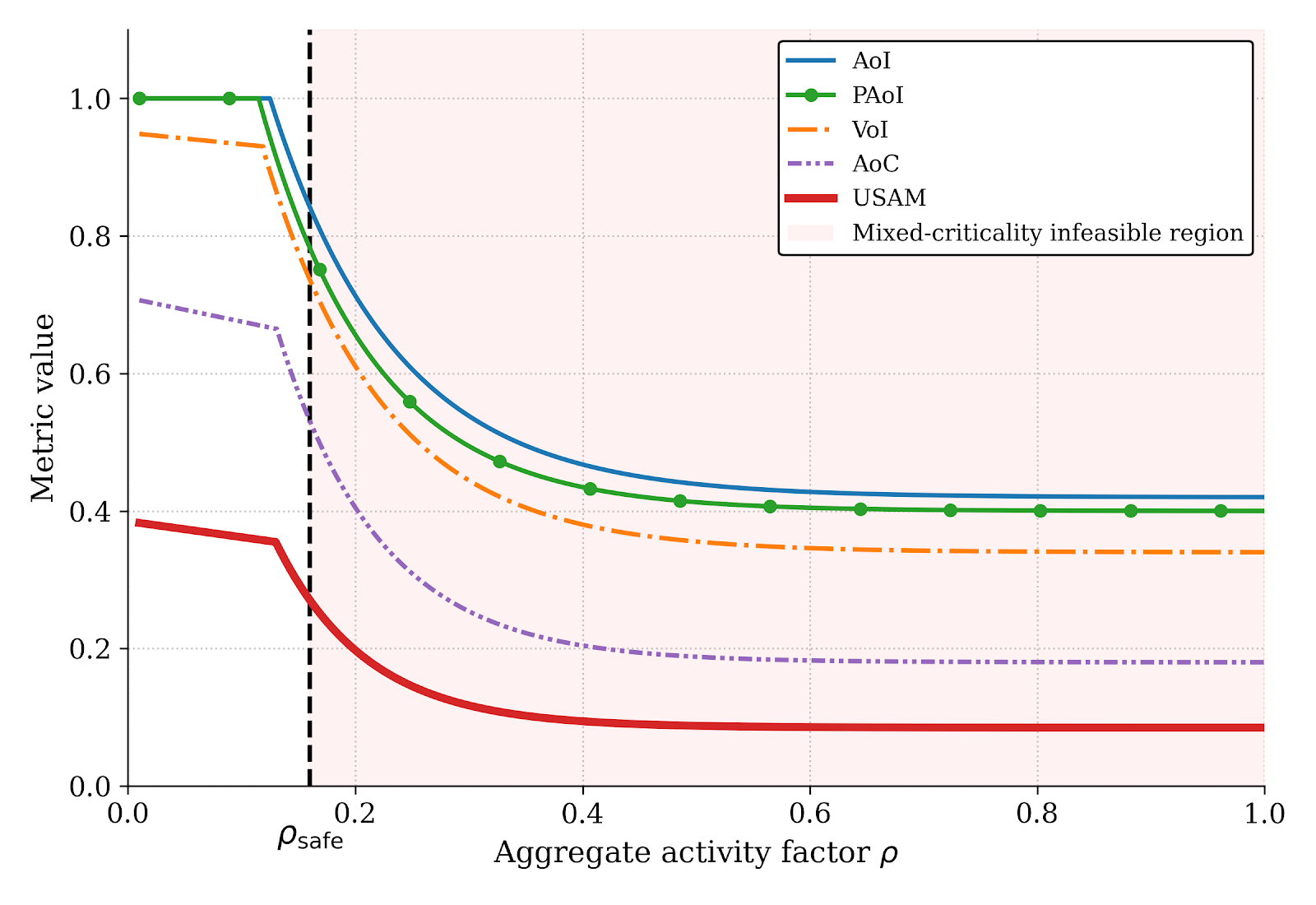}
\caption{Metric values as a function of the aggregate activity factor $\rho$. The dashed line marks the feasibility threshold $\rho_{\mathrm{safe}}$. Classical timeliness metrics degrade smoothly with increasing load, while USAM reveals the mixed-criticality feasibility boundary.}
\label{fig:rho_sweep}
\end{figure}

The relationship between probabilistic reliability and deterministic timing is illustrated in Fig.~\ref{fig:prob_det}. Deadline reliability improves gradually with increasing receiver duty cycle because higher receiver availability reduces the probability of deadline violations. Deterministic safety feasibility depends on the worst-case response-time bound and therefore on the activation latency of the receiver. Increasing the duty cycle reduces this latency and improves deterministic safety timing. USAM combines these two effects, illustrating that high probabilistic reliability alone does not guarantee safe operation.

\begin{figure}[t]
\centering
\includegraphics[width=0.92\linewidth]{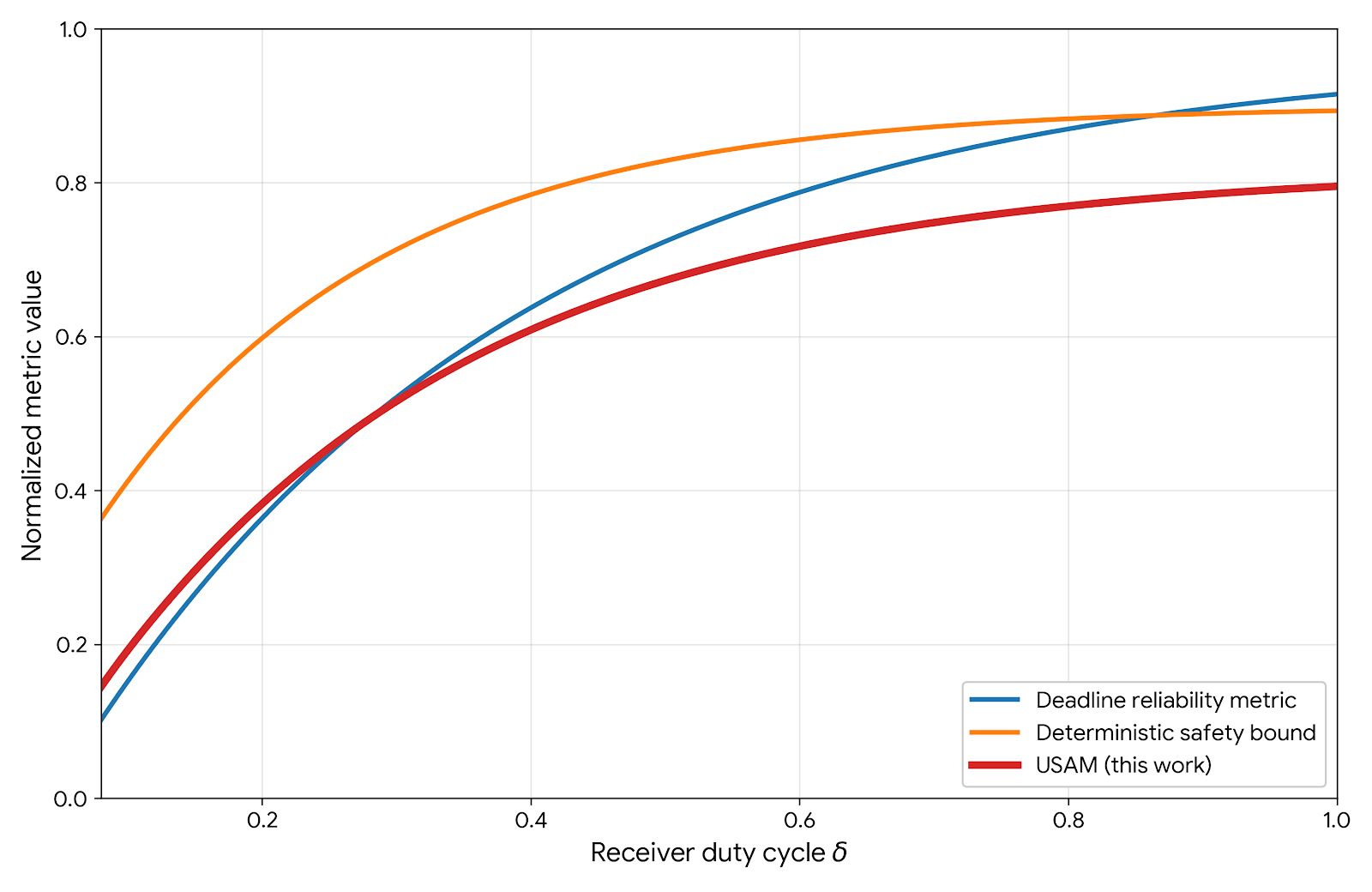}
\caption{Deadline reliability, deterministic safety feasibility, and the resulting USAM value as a function of receiver duty cycle $\delta$. Increasing $\delta$ improves both probabilistic reliability and deterministic timing guarantees.}
\label{fig:prob_det}
\end{figure}

The global feasibility structure derived in Section~IV is summarized in Fig.~\ref{fig:phase_diagram}. The dashed curve corresponds to the queue-stability boundary, while the vertical line indicates the deterministic safety threshold $\delta_{\mathrm{safe}}$. Three operating regimes appear. For $\delta<\delta_{\mathrm{safe}}$, the system is safety infeasible regardless of traffic load because worst-case response-time constraints cannot be satisfied. For $\delta\ge\delta_{\mathrm{safe}}$, the system becomes feasible for sufficiently small activity factors. As traffic load increases further, the system eventually enters a queue-limited regime in which classical queueing effects dominate performance.

\begin{figure}[t]
\centering
\includegraphics[width=0.92\linewidth]{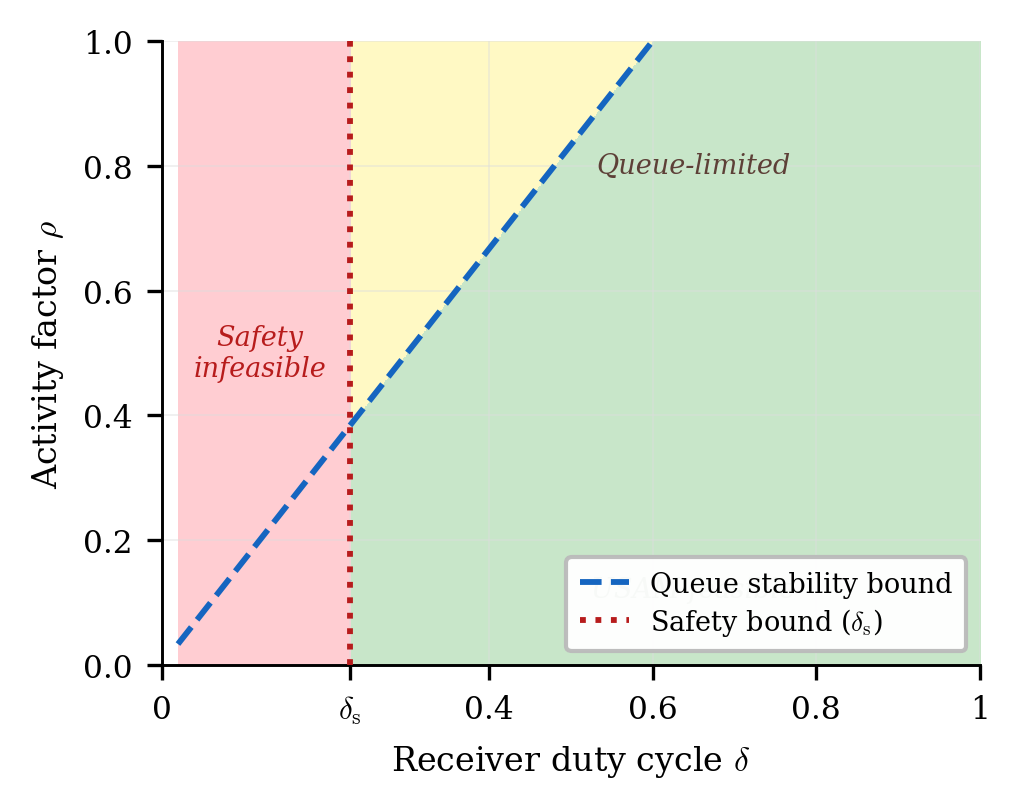}
\caption{Feasibility regions in the $(\delta,\rho)$ plane. The dashed curve shows the queue-stability boundary and the vertical line marks the safety duty-cycle threshold $\delta_{\mathrm{safe}}$.}
\label{fig:phase_diagram}
\end{figure}

The numerical results are consistent with the analytical results derived in Section~IV. In ultra-sparse IoT systems the feasibility of mixed-criticality communication is governed primarily by infrastructure readiness and deterministic safety timing rather than by average traffic load. Classical freshness-oriented metrics capture the evolution of information staleness but remain insensitive to deterministic safety feasibility, whereas the proposed USAM formulation reveals the architectural feasibility boundary of the system.
\section{Conclusion}

This paper introduced the Unified Safety--Age Metric (USAM) for evaluating timeliness in heterogeneous Internet-of-Things systems that simultaneously support monitoring, control, and safety-critical traffic. Unlike classical timeliness metrics such as Age of Information and its variants, which quantify information freshness or control performance, the proposed formulation explicitly incorporates deterministic safety feasibility through worst-case response-time constraints.

The analysis showed that in the ultra-sparse regime typical of massive IoT deployments, queueing effects vanish and system behavior becomes dominated by infrastructure readiness and deterministic timing constraints. In this regime, the feasibility of mixed-criticality communication is determined primarily by the receiver duty cycle rather than by aggregate traffic load. This leads to a simple design rule: the communication infrastructure should be dimensioned according to the most stringent safety and fast-control requirements rather than the average monitoring traffic.

Numerical results illustrate the structural difference between classical freshness-oriented metrics and the proposed formulation. While conventional metrics degrade smoothly as receiver availability decreases or traffic load increases, they do not reveal the deterministic feasibility boundary imposed by safety constraints. In contrast, USAM captures this boundary and identifies operating regions in which heterogeneous traffic requirements cannot be satisfied simultaneously.

The proposed framework provides a basis for analyzing safety-aware communication architectures in large-scale IoT systems. Future work includes protocol-level realizations and adaptive infrastructure control mechanisms that jointly optimize freshness, reliability, and deterministic safety guarantees.


\end{document}